\documentclass[draftclsnofoot,onecolumn,12pt]{IEEEtran}

\usepackage[T1]{fontenc}
\usepackage{amsthm}
\usepackage{amsmath}
\usepackage[cmintegrals]{newtxmath}
\usepackage{mathtools,amssymb,bm}
\usepackage{graphicx}
\usepackage{color}
\usepackage{multicol}
\usepackage[font={footnotesize}]{caption}

\usepackage{lettrine}
\usepackage{footnote}
\makesavenoteenv{tabular}
\usepackage[noend]{algpseudocode}

\makeatletter
\newcommand{\Vast}{\bBigg@{5}}
\newcommand{\vast}{\bBigg@{3}}
\makeatother

\renewcommand\IEEEkeywordsname{Index Terms}
\renewcommand{\IEEEQED}{$\Box$}

\newcommand{\dB}{\,\mathrm{dB}}
\newcommand{\nm}{\,\mathrm{nm}}
\newcommand{\E}{\,\mathrm{E}}

\newcommand{\var}{\,\mathrm{Var}}
\newcommand{\argmax}{\,\mathrm{argmax}}
\newcommand{\matr}[1]{\mathbf{#1}}

\theoremstyle{definition}

\newtheorem{theo}{\textbf{Theorem}}
\newtheorem{lem}{\textbf{Lemma}}
\newtheorem{cor}{\textbf{Corollary}}
\newtheorem{rem}{\textit{Remark}}

\usepackage{siunitx}
\makeatletter
\providecommand\add@text{}
\newcommand\tagaddtext[1]{%
  \gdef\add@text{#1\gdef\add@text{}}}%
\renewcommand\tagform@[1]{%
  \maketag@@@{\llap{\add@text\quad}(\ignorespaces#1\unskip\@@italiccorr)}%
}
\makeatother

\begin{document}

\title{\LARGE{Jamming Suppression in Massive MIMO Systems}}

\author{\normalsize{Hossein Akhlaghpasand, Emil Bj\"{o}rnson, and S. Mohammad Razavizadeh}
\thanks{\IEEEauthorblockA{ H. Akhlaghpasand and S. M. Razavizadeh are with Iran University of Science and Technology, Tehran, Iran (e-mail: h{\_}akhlaghpasand@elec.iust.ac.ir, smrazavi@iust.ac.ir). E. Bj\"{o}rnson is with Link\"{o}ping University, Link\"{o}ping, Sweden (e-mail: emil.bjornson@liu.se).}}}

\maketitle

\begin{abstract}
In this paper, we propose a framework for protecting the uplink transmission of a massive multiple-input multiple-output (mMIMO) system from a jamming attack. Our framework includes a novel minimum mean-squared error based jamming suppression (MMSE-JS) estimator for channel training and a linear zero-forcing jamming suppression (ZFJS) detector for uplink combining. The MMSE-JS exploits some intentionally unused pilots to reduce the pilot contamination caused by the jammer. The ZFJS suppresses the jamming interference during the detection of the legitimate users' data symbols. The implementation of the proposed framework is practical, since the complexities of computing the MMSE-JS and the ZFJS are linear (not exponential) with respect to the number of antennas at the base station and linear detectors with the same complexities as the ZFJS have been already fabricated using $28~\nm$ FD-SOI (Fully Depleted Silicon On Insulator) technology in \cite{MaMIMOImpLarsson} and Xilinx Virtex-7 XC7VX690T FPGA in \cite{MaMIMOImpZeng} for the mMIMO systems. Our analysis shows that the jammer cannot dramatically affect the performance of a mMIMO system equipped with the combination of MMSE-JS and ZFJS. Numerical results confirm our analysis.
\end{abstract}


\section{Introduction}
Massive multiple-input multiple-output (mMIMO) is the key technology to increase the spectral efficiency (SE) in future wireless networks \cite{BookMarzetta}, \cite{MassiveMIMOPaper}, by spatial multiplexing of many users. Robustness against jamming attacks is an important requirement that future networks must fulfill. Jamming detection is the first step to improve security \cite{JamDetHoss}, while jamming suppression is the next step \cite{AntiJamTransVec}$-$\cite{JamResTai}. In \cite{AntiJamTransVec}, two anti-jamming schemes are proposed for an interference alignment (IA)-based wireless network with perfect channel state information and cooperation between the transmitters and receivers, while in \cite{HarvestJamTransWireless}, exploiting jamming signals for energy harvesting is also considered in such a network to take benefit from the jammer. In \cite{JamMitDet}, a subspace-based approximate minimum mean-squared error (MMSE) detector is used to filter out uplink jamming, but it breaks down when the jammer uses the same power as the legitimate users. Since mMIMO is sensitive to jamming pilot contamination, in \cite{SecTransWang}, the downlink transmission is secured by assigning multiple orthogonal pilots to each user and randomly selecting which to transmit. Do \emph{et al.} \cite{JamResTai} proposed to estimate the jammer channel and use it for zero-forcing (ZF)-like uplink detection. The jamming channel is estimated by exploiting unused pilots and pilot hopping, to ``force'' the jammer to transmit these unused pilots to cause pilot contamination. We will improve on this prior work by supporting multiple legitimate users and using the unused pilots in a more efficient way to achieve higher SE.

The severe threats imposed by jamming attacks during the uplink transmission of the mMIMO systems motivated us to study the jamming suppression problem. The main contributions of this paper can be summarized as follows:
\begin{itemize}
\item We suggest a new uplink framework based on a zero-forcing jamming suppression (ZFJS) detector, which uses the channels of both the legitimate users and jammer to suppress jamming.
\item We propose an MMSE based jamming suppression (MMSE-JS) estimator that exploits unused pilots to provide the estimated channels of both the legitimate users and jammer for implementation of ZFJS.
\item We exploit asymptotic properties of mMIMO to acquire the channel statistics and utilize pilot hopping to reduce the jamming pilot contamination in order to get accurate channel estimations.
\end{itemize}
Our analytical and numerical analyses show that using the proposed framework, the jammer cannot dramatically affect SE of the system.

\textit{Notations:} The conjugate, transpose, and conjugate transpose of an arbitrary matrix $\matr{X}$ are respectively denoted by $\matr{X}^*$, $\matr{X}^T$, and $\matr{X}^H$, while $\matr{I}_N$ is the $N \times N$ identity matrix. We denote the Euclidean norm of an arbitrary vector $\matr{x}$ by $\|\matr{x} \|$. The magnitude and angle of a complex value $x$ are denoted by $|x|$ and $\angle x$, respectively. $\text{E} \{\matr{x} \}$ is the expectation of a stochastic vector $\matr{x}$. Finally, $\var \{x\} \triangleq \text{E}\{|x-\text{E}\{x\}|^2\}$.

\section{Problem Setup}
We consider the uplink of a single-cell mMIMO system, depicted in Fig. \ref{fig1}, consisting of one $M$-antenna base station (BS) and $K$ single-antenna legitimate users in the presence of a jammer. For brevity, we consider the single-cell scenario, which is similar to a practical multi-cell scenario with a pilot reuse factor $3$, $4$, or $7$. We assume a single-antenna jammer in order to analyze the main principle of jamming suppression within the limited space. We denote by $\matr{h}_{k} \in \mathbb{C}^{M \times 1}$ and $\matr{h}_w \in \mathbb{C}^{M \times 1}$ the channel vectors from the $k$th user and the jammer to the BS, respectively, which are modeled by $\matr{h}_{k} \sim \mathcal{CN} \left(\matr{0} , \beta_{k} \matr{I}_M \right)$ and $\matr{h}_w \sim \mathcal{CN} \left(\matr{0} , \beta_w \matr{I}_M \right)$. The coefficients $\beta_{k}$ and $\beta_w$ represent the large-scale fading. We study a block-fading model where the channels are constant within a coherence block of $T$ samples and have independent realizations between the blocks.

The BS needs to estimate the channel vectors of the users under interference/noise in each coherence block. Hence, the users transmit their mutually orthonormal pilots at $\tau$ ($\tau \leq T$) samples of a coherence block. The pilots are selected from an orthonormal pilot set and assigned to the users by the BS. We denote the pilot of the $k$th user by $\bm{\phi}_k^{\left(\text{U}\right)} \in \mathbb{C}^{\tau \times 1}$ and the orthonormal pilot set by $\left \{\bm{\phi}_1, \ldots, \bm{\phi}_{\tau} \right \}$, where $\bm{\phi}_i \in \mathbb{C}^{\tau \times 1}$ is the $i$th pilot. During the remaining $T-\tau$ samples, the users transmit data symbols to the BS. The transmit powers of the users are denoted by $p_t$ and $p_d$ in the pilot and data phases, respectively. The jammer transmits intentional interference in the pilot and data phases by different transmit powers $q_t$ and $q_d$, respectively, in order to reduce the performance of the legitimate system.

\begin{figure}
\centering
\includegraphics[width=0.4\columnwidth]{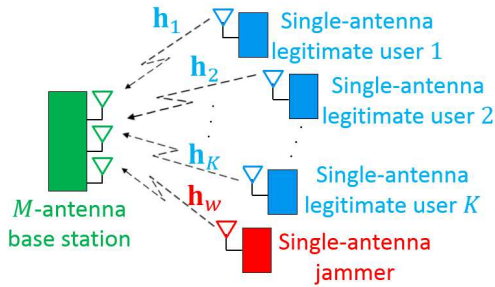}
\caption{Uplink transmission of a mMIMO system under jamming attack.}
\label{fig1}
\end{figure}

\subsection{Pilot Phase}
During $\tau$ samples of a coherence block, the users transmit their allocated pilots to the BS. The received signal $\matr{Y}_t \in \mathbb{C}^{M \times \tau}$ at the BS is given by
\begin{equation}\label{received_matrix}
\matr{Y}_t = \sqrt{\tau p_t} \sum_{i=1}^{K}\matr{h}_{i} \bm{\phi}^{\left(\text{U}\right)^T}_{i} + \sqrt{\tau q_t} \matr{h}_w \bm{\psi}^T_{w} + \matr{N}_t ,
\end{equation}
where $\bm{\psi}_w \in \mathbb{C}^{\tau \times 1}$ is a normalized (unit power) sequence which is transmitted by the jammer and $\matr{N}_t \in \mathbb{C}^{M \times \tau}$ is the normalized receiver noise composed of i.i.d $\mathcal{CN}(0,1)$ elements. Since $\bm{\psi}_w$ might be a random vector, the second term in the right-hand side of \eqref{received_matrix} is not necessarily Gaussian. Hence, we use the linear MMSE (LMMSE) estimate of ${\matr{h}}_k$ based on $\matr{Y}_t$ which takes the following form \cite[Ch. 12]{BookEstKay}:
\begin{equation}\label{Conv_MMSE}
\hat{\matr{h}}_k = \frac{\sqrt{\tau p_t} \beta_k}{1+\tau q_t \E \left\{\left|\alpha_k^{\left(\text{U}\right)} \right|^2 \right\} \beta_w+\tau p_t \beta_k} \matr{y}^{\left(\text{U}\right)}_{k} ,
\end{equation}
where $\alpha_k^{\left(\text{U}\right)} \triangleq \bm{\psi}^T_{w} \bm{\phi}^{\left(\text{U}\right)^*}_{k}$ and $\matr{y}^{\left(\text{U}\right)}_{k} \triangleq \matr{Y}_t \bm{\phi}^{\left(\text{U}\right)^*}_{k}$. To implement the LMMSE, we need $M\tau+7$ complex multiplications.

\subsection{Data Phase}
During data transmission, the normalized complex symbols $x_i$ and $x_w$ are simultaneously transmitted by the $i$th user and jammer. The received vector $\matr{y}_d \in \mathbb{C}^{M \times 1}$ at the BS is
\begin{equation}\label{rec_sig_data_phase}
\matr{y}_d = \sqrt{p_d} \sum_{i=1}^{K} \matr{h}_{i} x_{i} + \sqrt{q_d} \matr{h}_w x_w + \matr{n}_d ,
\end{equation}
where $\matr{n}_d \sim \mathcal{CN}\left(\matr{0}, \matr{I}_M\right)$ is the normalized receiver noise. The BS filters the data signal $\matr{y}_d$ by using a combining matrix $\matr{V} \in \mathbb{C}^{M \times K}$, i.e., the receive combining matrix based on the channel estimates, to obtain $\matr{V}^H \matr{y}_d$. 
\begin{lem} \label{Lemma1}
(Asymptotic SE): In a mMIMO system utilizing the LMMSE estimate in \eqref{Conv_MMSE} and a standard linear detector $\matr{V}$ (e.g., matched filter (MF) or ZF), if the probability of $\alpha_k^{\left(\text{U}\right)}=0$ is non-zero, then the achievable SE of the $k$th user goes to infinity as $M \rightarrow \infty$; otherwise if the probability of $\alpha_k^{\left(\text{U}\right)}=0$ is zero, then we get the asymptotic SE of the $k$th user as
\begin{equation}
\label{asymp_SE}
\mathcal{S}_k^{\left(\text{asy}\right)} = \left(1-\frac{\tau}{T}\right) \E \left\{\log_2 \left(1+\frac{p_t p_d \beta_k^2}{q_t q_d \left|\alpha_k^{\left(\text{U}\right)} \right|^2 \beta_w^2} \right) \right\} , M \rightarrow \infty .
\end{equation}
\end{lem}
\begin{proof}
Please refer to Appendix \ref{PLemma1}.
\end{proof}
Since in the mMIMO systems, the asymptotic SE tends to infinity, an effective jammer tries to limit the asymptotic SE in \eqref{asymp_SE} (or even makes it to be zero). The necessary condition to achieve this goal is that $\alpha_k^{\left(\text{U}\right)}$ is always non-zero. Moreover, in a special worst-case situation, a jammer may have access to the pilot assigned to a particular user and attack that user (e.g., the $k$th user). In this case, the jammer selects $|\alpha_k^{\left(\text{U}\right)}|^2=1$. However, this is not the case that happens in reality due to that in a hostile environment, the BS knows such a jammer attacks the pilot phase, thus it uses a pseudo-random pilot hopping technique to assign pilots to the users at the beginning of each coherence block as a counter strategy \cite[Sec. V]{SecureMaMIMO}. Therefore, the jammer cannot \textbf{immediately} estimate which pilot is assigned to the target user nor exactly transmit the same pilot at each instant. Anyway, we assume that the jammer is smart so that it knows the transmission protocol and the pilot set $\left \{\bm{\phi}_1, \ldots, \bm{\phi}_{\tau} \right \}$. A smart jammer can be aware of this information by listening to the channel for some consecutive coherence blocks.
\begin{rem}
\textit{In a mMIMO system, the main aim of a jammer can be to limit the asymptotic SE of a system or a particular user. Since the jammer cannot know which pilot is currently used by the target user, to ensure effective attacking, it selects $\bm{\psi}_{w}$ such that the inner product of $\bm{\psi}_{w}$ and each $\bm{\phi}_{i}$ in the pilot set is always non-zero. The jammer needs to spread its power over all pilots, i.e., $0 < \E\{|\bm{\psi}_{w}^T \bm{\phi}_{i}^*|^2\} < 1,~i=1,\ldots,\tau$, to make the probability of $\alpha_k^{\left(\text{U}\right)}=0$ to be zero. Prior works assume that the jammer divides its average power equally over all pilots, i.e., $\E\{|\bm{\psi}_{w}^T \bm{\phi}_{i}^*|^2\} = 1/\tau$, that guarantees any jamming sequence is probabilistically equally good \cite{JamResTai,SecureMaMIMO}.}
\end{rem}
\textbf{Example:} An example of a jammer strategy that satisfies above description is
\begin{equation}
\bm{\psi}_w \sim \mathcal{CN}\left(\matr{0}, \frac{1}{\tau}\matr{I}_{\tau}\right) .
\end{equation}
We can compute
\begin{equation}
\E \left\{\left|\bm{\psi}_{w}^T \bm{\phi}_{i}^*\right|^2 \right\} = \frac{1}{\tau} ,
\end{equation}
because the pilots are normalized, i.e., $\bm{\phi}_{i}^T \bm{\phi}_{i}^* = 1$. The probability of $\bm{\psi}_{w}^T \bm{\phi}_{i}^* = 0$ is zero, since $\bm{\psi}_{w}$ is a continuous random vector, which is uniformly distributed over the unit sphere. \hspace{2.35cm} $\Box$

The number of pilots is fixed in practice, while the number of users vary and is usually less than $\tau$. Hence, there are generally $\tau-K \geq 1$ unused pilots which can be exploited to improve the security.

\section{Jamming Suppression Framework}
In this section, we propose the MMSE-JS estimator as the first step of our framework. Then, we use a ZFJS detector for next step and mathematically analyze the SE of the system.

\subsection{MMSE-JS Estimator}
To simplify the notation, we assume without loss of generality that pilots used in one coherence block are numbered such that the first $K$ pilots are assigned to the users and the remaining pilots are unused. The received signal in \eqref{received_matrix} is rewritten as
\begin{equation}\label{received_mat_simpler}
\matr{Y}_t = \sqrt{\tau p_t} \sum_{i=1}^{K}\matr{h}_{i} \bm{\phi}^T_{i} + \sqrt{\tau q_t} \matr{h}_w \bm{\psi}^T_{w} + \matr{N}_t .
\end{equation}
Since the signals received along the unused pilots include only the jamming signal (and noise), we exploit them to estimate the jammer's channel and improve the estimation of the users' channels. By projecting $\matr{Y}_t$ on each of the unused pilots, we have
\begin{equation}\label{unused_pilot_sig}
\matr{y}_{i} = \matr{Y}_t \bm{\phi}_{i}^* = \sqrt{\tau q_t} {\alpha}_i \matr{h}_w + \matr{n}_{i} ,
\quad i=K + 1,\ldots, \tau ,
\end{equation}
where $\alpha_i \triangleq \bm{\psi}^T_{w} \bm{\phi}^*_{i}$ and $\matr{n}_{i}=\matr{N}_t \bm{\phi}_{i}^* \sim \mathcal{CN} ( \matr{0} , \matr{I}_M )$. The MMSE-JS estimate of the jammer's effective channel $\tilde{\matr{h}}_{w_i} \triangleq \sqrt{\tau q_t} {\alpha}_i \matr{h}_w$ based on $\matr{y}_{i}$, which is derived from MMSE estimation \cite[Ch. 11]{BookEstKay}, is
\begin{equation}\label{Jam_est}
\hat{\tilde{\matr{h}}}_{w_i} = \frac{\mathcal{J}_i}{1 + \mathcal{J}_i} \matr{y}_{i} = \frac{\sqrt{\tau q_t} {\alpha}_i \mathcal{J}_i}{1 + \mathcal{J}_i} \matr{h}_w + \frac{\mathcal{J}_i}{1 + \mathcal{J}_i} \matr{n}_{i} ,
\end{equation}
where $\mathcal{J}_i \triangleq \tau q_t \left|{\alpha}_i\right|^2 \beta_w$. The estimation error in this case is calculated as $\tilde{\bm{\varepsilon}}_{w_i} = \hat{\tilde{\matr{h}}}_{w_i} - \tilde{\matr{h}}_{w_i}$, where $\tilde{\bm{\varepsilon}}_{w_i}$ is independent of $\hat{\tilde{\matr{h}}}_{w_i}$ and $\tilde{\bm{\varepsilon}}_{w_i} \sim \mathcal{CN} \left(\matr{0} , \frac{\mathcal{J}_i}{1 + \mathcal{J}_i}\matr{I}_M \right)$. Although $\mathcal{J}_i$ is unknown by the system, we will later show that it can be estimated using the asymptotic behaviour of mMIMO systems. We have $\tau - K$ estimated versions of jammer's effective channels $\hat{\tilde{\matr{h}}}_{w_i} , ~ i=K+1,\ldots,\tau$, therefore we select the one that has the lowest normalized estimation error, i.e., $\tilde{\bm{\varepsilon}}_{w_i}^{\left(n\right)} \triangleq \frac{\tilde{\bm{\varepsilon}}_{w_i}}{\sqrt{\tau q_t} \alpha_i}$. Hence, the estimated effective channel of the jammer is obtained as $\hat{\tilde{\matr{h}}}^{\star}_{w} = \frac{\mathcal{J}_{\text{o}}\matr{y}_{\text{o}}}{1 + \mathcal{J}_{\text{o}}}$ where the subscript $\text{o} = \underset{i}{\argmax}\left\{\mathcal{J}_i \right\}$. The related estimation error is denoted by $\tilde{\bm{\varepsilon}}^{\star}_{w}$. On the other hand, the projection of $\matr{Y}_t$ on the $k$th user's pilot sequence is
\begin{equation}\label{proj_est_act}
\matr{y}_{k} = \matr{Y}_t \bm{\phi}_{k}^* = \sqrt{\tau p_t} \matr{h}_{k} + \sqrt{\tau q_t} {\alpha}_k \matr{h}_w + \matr{n}_{k} ,
\end{equation}
where $\alpha_k \triangleq \bm{\psi}^T_{w} \bm{\phi}^*_{k}$ and $\matr{n}_{k}=\matr{N}_t \bm{\phi}_{k}^* \sim \mathcal{CN} \left( \matr{0} , \matr{I}_M \right)$. We want to estimate the channel of the $k$th user based on $\matr{y}_k$. To decrease the pilot contamination caused by the jammer, we need to suppress the jammer's term in \eqref{proj_est_act}. To this end, we exploit one of the received signals on the unused pilots. We select this signal from \eqref{unused_pilot_sig} and denote it by $\matr{y}_{\bar{k}}$. The MMSE-JS estimate of the $k$th user's channel based on $\matr{y}_k$ and $\matr{y}_{\bar{k}}$ is
\begin{equation}\label{leg_est}
\hat{\matr{h}}_{k} = \frac{\sqrt{\tau p_t} \beta_{k}}{1 + \delta_{k}^2 + \tau p_t \beta_{k}} \left(\matr{y}_{k} - \delta_{k} e^{j \theta_{k}} \matr{y}_{\bar{k}} \right) = \frac{\tau p_t \beta_{k}}{1 + \delta_{k}^2 + \tau p_t \beta_{k}} \matr{h}_{k} + \frac{\sqrt{\tau p_t} \beta_{k}}{1 + \delta_{k}^2 + \tau p_t \beta_{k}} \left(\matr{n}_k - \delta_{k} e^{j \theta_{k}} \matr{n}_{\bar{k}} \right) ,
\end{equation}
where $\delta_{k} \triangleq \left|\frac{\alpha_{k}}{\alpha_{\bar{k}}} \right|$ and $\theta_{k} \triangleq \angle \alpha_{k} - \angle \alpha_{\bar{k}}$. The estimation error is independent of $\hat{\matr{h}}_{k}$ and given by $\bm{\varepsilon}_{k} \sim \mathcal{CN} \left( \matr{0} , \frac{\left(1 + \delta_{k}^2 \right) \beta_{k}}{1 + \delta_{k}^2 + \tau p_t \beta_{k}} \matr{I}_M \right)$. The BS does not know $\delta_{k}$ and $\theta_{k}$, but it can estimate them using the asymptotic mMIMO behaviors, as shown next.

\subsection{Asymptotic Estimates}
To implement MMSE-JS, the BS needs to know the variables $\mathcal{J}_i , ~ i = K+1,\ldots, \tau$ and $\delta_{k} , \theta_{k} , ~ k = 1,\ldots,K$, which are generated randomly by the jammer. Recall that the subscript $\bar{k}$ in the definition of $\delta_{k}$ and $\theta_{k}$ refers to the received signal on the unused pilot selected from \eqref{unused_pilot_sig} to suppress the jammer's pilot contamination from the estimate of the $k$th user's channel.
\begin{theo} \label{Theo1}
In mMIMO systems, $\mathcal{J}_i$ can be estimated by
\begin{equation}
\label{asymp_MaMIMO}
\mathcal{J}_{i} = \left[\frac{\left\|\matr{y}_i \right\|^2}{M} - 1 \right]^+ , ~ i = K+1,\ldots,\tau, ~ M \rightarrow \infty ,
\end{equation}
where $\matr{y}_i$ is the received signals in \eqref{unused_pilot_sig}, and $[x]^+ \triangleq \max (0,x)$ for $x \in \mathbb{R}$. In addition, $\delta_{k} , \theta_{k}$ are obtained as
\begin{align}
\delta_{k} & = \Vast \{\begin{array}{lr}
\sqrt{\frac{\left[\frac{\left\|\matr{y}_k \right\|^2}{M} - \tau p_t \beta_{k} - 1\right]^+}{\frac{\left\|\matr{y}_{\bar{k}} \right\|^2}{M} - 1}} , & \frac{\left\|\matr{y}_{\bar{k}} \right\|^2}{M} > 1 , \\
0 , & \frac{\left\|\matr{y}_{\bar{k}} \right\|^2}{M} \leq 1 ,
\end{array} ~ M \rightarrow \infty , \label{asymp_MaMIMO2} \\
\theta_{k} & = \angle \frac{\matr{y}_{\bar{k}}^H \matr{y}_{k}}{M} , ~ M \rightarrow \infty \label{asymp_MaMIMO3},
\end{align}
for $k=1,\ldots,K$ following the received signal on the $k$th user's pilot $\matr{y}_k$ in \eqref{proj_est_act} as well as $\matr{y}_{\bar{k}}$ from \eqref{unused_pilot_sig}.
\end{theo}
\begin{proof}
Please refer to Appendix \ref{PTheorem1}.
\end{proof}
The number of complex multiplications required to implement the MMSE-JS is $M\left(3\tau+4\right)+16$, which is linear with respect to $M$ and $\tau$ (similar to the complexity of the LMMSE). Hence, the implementation of the MMSE-JS is practical.

\subsection{ZFJS Detector}
We exploit ZFJS based on the estimated channels of the users $\hat{\matr{h}}_{k},~k=1,\ldots,K$ and the estimated channel of the jammer $\hat{\tilde{\matr{h}}}^{\star}_w$. The ZFJS is defined as $\matr{V} \triangleq \matr{H}(\matr{H}^H \matr{H})^{-1}$, where $\matr{H} \triangleq [\hat{\matr{h}}_{1},\ldots,\hat{\matr{h}}_{K} , \hat{\tilde{\matr{h}}}^{\star}_w ] \in \mathbb{C}^{M \times (K+1)}$ and $M > K+1$. We null the interference caused by the jammer and also other users using the ZFJS detector. The number of complex multiplications for computing the ZFJS is $M(K+1)^2+((K+1)^3-(K+1))/3$ as well as  the linear detectors with the same
complexities as the ZFJS have been already fabricated using $28~\nm$ FD-SOI (Fully Depleted Silicon On Insulator) technology in \cite[Sec. V]{MaMIMOImpLarsson} and Xilinx Virtex-7 XC7VX690T FPGA in \cite[Sec. VI]{MaMIMOImpZeng} for $M=128$ and $K=8$.
\begin{lem} \label{Lemma2}
Using ZFJS method based on the channels estimated by MMSE-JS, the achievable SE of the $k$th user is
\begin{equation}
\label{sum_se_imperfect}
\mathcal{S}_k = \left(1-\frac{\tau}{T}\right) \E \left\{\log_2 \left(1+{\rho}_{k} \right) \right\} ,
\end{equation}
where ${\rho}_{k}$ the effective SINR of the $k$th user is given by
\begin{equation}\label{sinr_imperfect}
{\rho}_{k} = \frac{p_d}{\var \left\{\sqrt{p_d} \sum_{i=1}^{K} \matr{v}_k^H \bm{\varepsilon}_{i} + \sqrt{\frac{q_d}{\mathcal{J}_{\text{o}}}} \matr{v}_k^H \tilde{\bm{\varepsilon}}^{\star}_w \Big| \alpha_k \right\} + \E \left\{\left\|\matr{v}_k \right\|^2 \Big| \alpha_k \right\}} ,
\end{equation}
and $\matr{v}_k$ is the $k$th column of the ZFJS matrix $\matr{V}$.
\end{lem}
\begin{proof}
Please refer to Appendix \ref{PLemma2}.
\end{proof}
\begin{cor}\label{cor1}
For $\tau \geq 2K+1$, the SINR in \eqref{sinr_imperfect} becomes
\begin{equation}\label{sinr_special_case}
{\rho}_{k} = \frac{\frac{\tau p_t p_d \beta_{k}^2}{1 + \delta_{k}^2 + \tau p_t \beta_{k}} \left(M-K-1 \right)}{\sum_{i=1}^{K} \frac{p_d \left(1 + \delta_{i}^2 \right) \beta_{i}}{1 + \delta_{i}^2 + \tau p_t \beta_{i}} + \frac{q_d}{1 + \mathcal{J}_{o}} + 1} .
\end{equation}
\end{cor}
\begin{proof}
When $\tau \geq 2K+1$, we can use different $\matr{y}_{\bar{k}}$ in \eqref{leg_est} for all users and also a different $\matr{y}_{\text{o}}$ for the jammer, therefore the estimation errors are mutually independent. Hence, $\var \{\sqrt{p_d} \sum_{i=1}^{K} \matr{v}_k^H \bm{\varepsilon}_{i} + \sqrt{\frac{q_d}{\mathcal{J}_{\text{o}}}} \matr{v}_k^H \tilde{\bm{\varepsilon}}^{\star}_{w} | \alpha_k \} = 
(\sum_{i=1}^{K} \frac{p_d \left(1 + \delta_{i}^2 \right) \beta_{i}}{1 + \delta_{i}^2 + \tau p_t \beta_{i}} + \frac{q_d}{1 + \mathcal{J}_{\text{o}}}) \text{E} \{\left\|\matr{v}_k \right\|^2 | \alpha_k \}$. The proof is complete after simplifying $\text{E} \{\|\matr{v}_k \|^2 | \alpha_k \}$ by using Wishart matrix properties \cite[Lemma 2.10]{BookTulino}.
\end{proof}
\begin{rem}
\textit{From \eqref{sinr_special_case}, ${\rho}_{k} \rightarrow \infty$ as $M \rightarrow \infty$. This means that the jammer cannot dramatically affect SE of the $k$th user when the system is equipped with the MMSE-JS estimator and ZFJS detector.}
\end{rem}

\section{Numerical Results}
We numerically evaluate the sum-SE of a system that exploits MMSE-JS and ZFJS during the pilot and data phases, respectively. We consider $T=200$ samples per coherence block and $10^5$ independent runs for Monte-Carlo simulations. We normalize the large-scale fading coefficients $\beta_{i}=\beta_w=1$. Since the noise variance is previously normalized, the transmit powers can be interpreted as the signal-to-noise ratios. Theorem \ref{Theo1} is used for implementation of the MMSE-JS estimator.

\begin{figure}
\centering
\includegraphics[width=0.5\columnwidth]{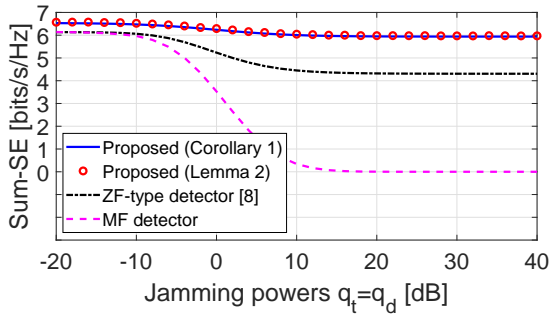}
\caption{Sum-SE of the system versus the jamming powers $q_t=q_d$ for $M = 100$, $K = 1$, $\tau = 3$, $p_t = p_d = 5 \dB$.}
\label{fig2}
\end{figure}
Fig.~\ref{fig2} depicts how the sum-SE varies with respect to the jamming powers for $M=100$, $K=1$, $\tau=3$, and $p_t=p_d=5 \dB$. The Monte-Carlo simulation of Lemma \ref{Lemma2} matches the closed-form expression in Corollary \ref{cor1}. Moreover, we compare the performance of the system equipped with the MMSE-JS estimator and the ZFJS detector, to the system equipped with the LMMSE estimator and the ZF-type detector proposed in \cite{JamResTai}. The sum-SE of a system equipped with the LMMSE estimator and MF detector is also illustrated to indicate the robustness level of the proposed framework against jamming attacks. This figure shows that the sum-SE with the proposed framework is substantially higher than in prior works. This happens since MMSE-JS decreases the pilot contamination caused by the jammer, while ZFJS exploits both the estimated user channels and the estimated channel of the jammer to detect the users' data symbols.

\begin{figure}
\centering
\includegraphics[width=0.5\columnwidth]{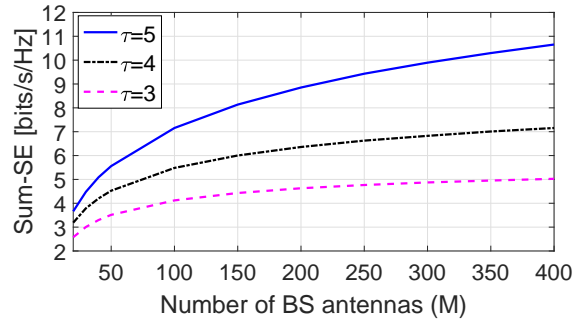}
\caption{Sum-SE of the system versus the number of BS antennas for $K = 2$, $\tau \in \{3,4,5\}$, and $p_t = p_d = q_t = q_d = 2 \dB$.}
\label{fig3}
\end{figure}
In Fig.~\ref{fig3}, the sum-SE is shown in terms of the number of BS antennas for $K = 2$, $\tau \in \{3,4,5\}$, and $p_t=p_d=q_t=q_d=2 \dB$. We can see that the performance of the system improves when the number of pilots increases. This is because $\matr{y}_{\bar{1}}$, $\matr{y}_{\bar{2}}$, and $\matr{y}_{\text{o}}$ in MMSE-JS can be selected to be different for larger $\tau$, therefore the correlation of the estimated channels decreases. For example for $\tau = 5$, we have three unused pilot signals, therefore each of them is selected as one of the signals $\matr{y}_{\bar{1}}$, $\matr{y}_{\bar{2}}$, and $\matr{y}_{\text{o}}$, subsequently the estimated channels are uncorrelated. We know that the linear detectors (like the ZFJS) perform better as the correlation of the estimated channels reduces.

\section{Conclusion}
In this paper, a framework was proposed to mitigate jamming during the uplink transmission of mMIMO systems. The proposed framework exploits the new MMSE-JS estimator and the ZFJS detector during the pilot and data phases, respectively. MMSE-JS uses unused pilots to reduce the pilot contamination caused by the jammer. ZFJS suppresses the jamming interference from the received signals at the BS during the data phase. Our analysis shows that the mMIMO system equipped with MMSE-JS and ZFJS is robust against jamming attacks even with strong jamming powers.

%

\appendix
\subsection{Proof of Lemma 1}\label{PLemma1}
Based on $\matr{v}_k^H \matr{y}_d$ ($\matr{v}_k$ is the $k$th column of matrix $\matr{V}$ and $\matr{y}_d$ is the received vector in \eqref{rec_sig_data_phase}), an achievable SE of the $k$th user is \cite[Sec. 2.3.5]{BookMarzetta}
\setcounter{equation}{17}
\begin{equation}\label{ach_SE_kth_user}
\mathcal{S}_k = \left(1-\frac{\tau}{T} \right) \E \left\{\log_2 \left(1 + \rho_k \right) \right\} ,
\end{equation}
where the effective SINR of the $k$th user, $\rho_k$, is obtained as
\setcounter{equation}{18}
\begin{equation}\label{eff_sjnr_kth_user}
\rho_k = \frac{p_d \left|\text{E} \left\{\matr{v}_k^H \matr{h}_{k} \Big| \alpha_k^{\left(\text{U}\right)} \right\} \right|^2}{p_d \sum_{i=1}^{K} \text{E} \left\{\left|\matr{v}_k^H \matr{h}_{i} \right|^2 \Big| \alpha_k^{\left(\text{U}\right)} \right\} - p_d \left|\text{E} \left\{\matr{v}_k^H \matr{h}_{k} \Big| \alpha_k^{\left(\text{U}\right)} \right\} \right|^2 + q_d \text{E} \left\{\left|\matr{v}_k^H \matr{h}_{w} \right|^2 \Big| \alpha_k^{\left(\text{U}\right)} \right\} + \text{E} \left\{\left\|\matr{v}_k \right\|^2 \Big| \alpha_k^{\left(\text{U}\right)} \right\}} .
\end{equation}
The expectations in \eqref{eff_sjnr_kth_user} are computed by assuming the MF detector at the BS, i.e., $\matr{v}_k = \hat{\matr{h}}_k$, in closed form as follows:
\setcounter{equation}{19}
\begin{equation}\label{Nominator}
\text{E} \left\{\hat{\matr{h}}_k^H \matr{h}_{k} \Big| \alpha_k^{\left(\text{U}\right)} \right\} = BM\left(1+\tau q_t \E \left\{\left|\alpha_k^{\left(\text{U}\right)} \right|^2 \right\} \beta_w+\tau p_t \beta_k \right) ,
\end{equation}
\begin{equation}\label{second_order_moment_users}
\text{E} \left\{\left|\hat{\matr{h}}_{k}^H \matr{h}_{i} \right|^2 \Big| \alpha_k^{\left(\text{U}\right)} \right\} =\vast\{ \begin{array}{ll}
BM \left(\tau p_t M \beta_k^2 + \beta_k \left(1 + \tau q_t \left|\alpha_k^{\left(\text{U}\right)}\right|^2 \beta_w + \tau p_t \beta_k \right) \right) ,& i=k \\
BM \beta_i \left(1 + \tau q_t \left|\alpha_k^{\left(\text{U}\right)}\right|^2 \beta_w + \tau p_t \beta_k \right) ,& i \neq k
\end{array},
\end{equation}
\begin{equation}\label{second_order_moment_jam}
\text{E} \left\{\left|\hat{\matr{h}}_{k}^H \matr{h}_{w} \right|^2 \Big| \alpha_k^{\left(\text{U}\right)} \right\} = BM \left( \tau q_t \left|\alpha_k^{\left(\text{U}\right)} \right|^2 M \beta_w^2 + \beta_{w} \left(1 + \tau q_t \left|\alpha_k^{\left(\text{U}\right)}\right|^2 \beta_w + \tau p_t \beta_k \right) \right) ,
\end{equation}
\begin{equation}\label{rec_noise}
\text{E} \left\{\left\|\hat{\matr{h}}_k\right\|^2 \Big| \alpha_k^{\left(\text{U}\right)} \right\} = BM \left(1 + \tau q_t \left|\alpha_k^{\left(\text{U}\right)}\right|^2 \beta_w + \tau p_t \beta_k \right) ,
\end{equation}
where $B = \tau p_t \beta_k^2/(1+\tau q_t \E \{|\alpha_k^{\left(\text{U}\right)}|^2 \} \beta_w+\tau p_t \beta_k)^2$. Regarding the computations in \eqref{Nominator}$-$\eqref{rec_noise}, two terms with scaling factor $M^2$ remain in $\rho_k$ which are related to the desired channel and the jammer's interference, while the other terms scale as $M$. Hence, we have
\begin{equation}\label{asymp_sinr}
\rho_k^{\left(\text{asy}\right)} = \frac{p_t p_d \beta_k^2}{q_t q_d \left|\alpha_k^{\left(\text{U}\right)} \right|^2 \beta_w^2}~\text{as}~M \rightarrow \infty .
\end{equation}
The limit in \eqref{asymp_sinr} is also achievable for ZF and MMSE detectors \cite[Proposition 5]{NgoSpec}.

\subsection{Proof of Theorem 1}\label{PTheorem1}
To prove \eqref{asymp_MaMIMO}, we have
\begin{equation}
\label{ch_harden}
\frac{\left\|\matr{y}_i \right\|^2}{M} = \frac{1}{M} \left(\sqrt{\tau q_t} {\alpha}_i \matr{h}_w + \matr{n}_{i}\right)^H \left(\sqrt{\tau q_t} {\alpha}_i \matr{h}_w + \matr{n}_{i}\right) ,
\end{equation}
according \eqref{unused_pilot_sig}. There exist four terms in \eqref{ch_harden}, where two of them tend to zero as $M \rightarrow \infty$, i.e., $\frac{\matr{h}_w^H\matr{n}_{i}}{M} = \frac{\matr{n}_i^H\matr{h}_w}{M} = 0 ,$
and the terms $\tau q_t \left|\alpha_{i}\right|^2 \frac{\|\matr{h}_w\|^2 }{M} = \mathcal{J}_i$, $\frac{\|\matr{n}_i\|^2 }{M} = 1$. For finite $M$, it can happen that $\frac{\left\|\matr{y}_i \right\|^2}{M} < 1$, but then the pilot jamming attack is negligible. Hence, the negative values of $\mathcal{J}_i$ in \eqref{asymp_MaMIMO} are ignored and replaced by zero. Furthermore, $\delta_{k}$ in \eqref{asymp_MaMIMO2} is obtained analogous to the proof of $\mathcal{J}_i$. To prove \eqref{asymp_MaMIMO3}, following from \eqref{unused_pilot_sig}, \eqref{proj_est_act}, and the asymptotic behaviour of the mMIMO systems, all terms in $\frac{\matr{y}_{\bar{k}}^H \matr{y}_{k}}{M}$ tend to zero except
\begin{equation}\label{angle_proof_MaMIMO}
\frac{\matr{y}_{\bar{k}}^H \matr{y}_{k}}{M} = \tau q_t \left|\alpha_{k}\alpha_{\bar{k}}\right| e^{j \theta_{k}} \frac{\left\|\matr{h}_{w}\right\|^2}{M} = \tau q_t \left|\alpha_{k}\alpha_{\bar{k}}\right| \beta_{w} e^{j \theta_{k}} .
\end{equation}
We can see that the angle of $\frac{\matr{y}_{\bar{k}}^H \matr{y}_{k}}{M}$ tends to $\theta_{k}$ as $M \rightarrow \infty$.

\subsection{Proof of Lemma 2}\label{PLemma2}
An achievable SE of the $k$th user is obtained \cite[Section 2.3.5]{BookMarzetta} as
$\mathcal{S}_k = (1-\frac{\tau}{T}) \E \{\log_{2} (1 + \frac{p_d \left|\E\left\{\matr{v}_k^H \matr{h}_{k} | \alpha_k \right\} \right|^2}{\var \left\{\tilde{n}_{k} | \alpha_k \right\}}) \}$,
by the side information $\alpha_k$, where the desired signal is received over $\E \{\matr{v}_k^H \matr{h}_{k} | \alpha_k \}$ and the other terms in $\matr{v}_k^H \matr{y}_d$ are the effective noise $\tilde{n}_k$. Plugging the ZFJS, i.e., $\matr{v}_k^H \matr{h}_i=\bigg\{ \begin{array}{rl}
1 , &\quad i=k \\
0 , &\quad i \neq k
\end{array}$, into $\mathcal{S}_k$ and since $\matr{n}_d$ is independent of the channels, we obtain $\mathcal{S}_k$ as \eqref{sum_se_imperfect}.


\begin{thebibliography}{1}

\bibitem{BookMarzetta}
T. L. Marzetta, E. G. Larsson, H. Yang, and H. Q. Ngo, \newblock {\em Fundamentals of Massive MIMO}, Cambridge, NJ: Cambridge Univ. Press, 2016.

\bibitem{MassiveMIMOPaper}
L. Zhao, K. Li, K. Zheng, and M. Omair Ahmad, \newblock ``An analysis of the tradeoff between the energy and spectrum efficiencies in an uplink massive MIMO-OFDM system,'' \newblock{\em IEEE Trans. Circuits Syst. II, Exp. Briefs}, vol. 62, no. 3, pp. 291-295, Mar. 2015.

\bibitem{JamDetHoss}
H. Akhlaghpasand, S. M. Razavizadeh, E. Bj\"{o}rnson, and T. T. Do, \newblock ``Jamming detection in massive MIMO systems,'' \newblock {\em IEEE Wireless Commun. Lett.}, vol. 7, no. 2, pp. 242-245, Apr. 2018.

\bibitem{AntiJamTransVec}
N. Zhao, J. Guo, F. R. Yu, M. Li, and V. C. M. Leung, \newblock ``Anti-jamming schemes for interference alignment (IA)-based wireless networks,'' \newblock {\em IEEE Trans. Veh. Technol.}, vol. 66, no. 2, pp. 1271-1283, Feb. 2017.

\bibitem{HarvestJamTransWireless}
J. Guo, N. Zhao, F. R. Yu, X. Liu, and V. C. M. Leung, \newblock ``Exploiting adversarial jamming signals for energy harvesting in interference networks,'' \newblock {\em IEEE Trans. Wireless Commun.}, vol. 16, no. 2, pp. 1267-1280, Feb. 2017.

\bibitem{JamMitDet}
J. Vinogradova, E. Bj\"{o}rnson, and E. G. Larsson, \newblock ``Detection and mitigation of jamming attacks in massive MIMO systems using random matrix theory,'' in \newblock{\em Proc. IEEE SPAWC}, 2016, pp. 1-5.

\bibitem{SecTransWang}
H. M. Wang, K. W. Huang, and T. A. Tsiftsis, \newblock ``Multiple antennas secure transmission under pilot spoofing and jamming attack,'' \newblock{\em IEEE J. Sel. Areas Commun.}, vol. 36, no. 4, pp. 860-876, Apr. 2018.

\bibitem{JamResTai}
T. T. Do, E. Bj\"{o}rnson, E. G. Larsson, and S. M. Razavizadeh, \newblock ``Jamming-resistant receivers for massive MIMO uplink,'' \newblock{\em IEEE Trans. Inf. Forensics Security}, vol. 13, no. 1, pp. 210-223, Jan. 2018.

\bibitem{BookEstKay}
S. M. Kay, \newblock {\em Fundamentals of Statistical Signal Processing: Estimation Theory.} Englewood Cliffs, NJ: Prentice-Hall, 1993.

\bibitem{NgoSpec}
H. Q. Ngo, E. G. Larsson, and T. L. Marzetta, \newblock ``Energy and spectral efficiency of very large multiuser MIMO systems,'' \newblock {\em IEEE Trans. Commun.}, vol. 61, no. 4, pp. 1436-1449, Apr. 2013.

\bibitem{SecureMaMIMO}
Y. O. Basciftci, C. E. Koksal, and A. Ashikhmin, \newblock ``Securing massive MIMO at the physical layer,'' in \newblock{\em IEEE Conf. on Commun. and Net. Sec. (CNS)}, 2015, pp. 272-280.

\bibitem{MaMIMOImpLarsson}
L. Van der Perre, L. Liu, and E. G. Larsson, \newblock ``Efficient DSP and circuit architectures for massive MIMO: State of the art and future directions,'' \newblock {\em IEEE Trans. Signal Process.}, vol. 66, no. 18, pp. 4717-4736, Sept. 2018.

\bibitem{MaMIMOImpZeng}
J. Zeng, J. Lin, and Z. Wang, \newblock ``An improved Gauss-Seidel algorithm and its efficient architecture for massive MIMO systems,'' \newblock{\em IEEE Trans. Circuits Syst. II, Exp. Briefs}, vol. 65, no. 9, pp. 1194-1198, Sept. 2018.

\bibitem{BookTulino}
A. M. Tulino and S. Verd\'u, \newblock { \em Random Matrix Theory and Wireless Communications.} Foundations Trends Commun. and Inf. Theory, 2004.

\end{thebibliography}
\end{document}